\documentclass[a4paper,12pt]{article}
\usepackage[latin1]{inputenc}
\usepackage{amssymb}
\usepackage{lscape}
\usepackage{amsmath}
\usepackage{amsthm}
\usepackage{latexsym}
\usepackage{calc}
\usepackage{graphicx}
\usepackage{bm}
\usepackage{overpic}
  
\usepackage{exscale}
\usepackage{amsfonts}
\usepackage{amssymb}  
\usepackage[usenames]{color} 

 \def\ua{\uparrow}

 \def\wt{\widetilde}

\def\bR{\mathbb{R}}
\def\bT{\mathbb{T}}

\def\bR{\mathbb R}

\def\bN{\mathbb N}

\def\eps{\varepsilon}

\def\<{[}\def\>{]}

\newtheorem{theorem}{Theorem}
\newtheorem{proposition}[theorem]{Proposition}

\theoremstyle{definition}
\newtheorem{definition}[theorem]{Definition}

\newtheorem{remark}[theorem]{Remark}

\normalsize
 
\begin{document}
\title{\bf Model-free CPPI}
\author{ Alexander Schied\\
University of Mannheim\\
A5, 6\\
68131 Mannheim, Germany}
 
\date{\small May 25, 2013}

\maketitle

\begin{abstract}We consider Constant Proportion Portfolio Insurance (CPPI) and its dynamic extension, which may be called Dynamic Proportion Portfolio Insurance (DPPI). It is shown that these  investment strategies work within the    setting of F\"ollmer's pathwise It\^o calculus, which  makes no probabilistic assumptions whatsoever. This shows, on the one hand, that CPPI and DPPI are completely independent of any choice of a particular model  for the dynamics of asset prices. They  even make sense beyond the class of semimartingale sample paths and can be successfully defined for models admitting arbitrage, including some models based on fractional Brownian motion. On the other hand, the result can be seen as a case study for the general issue  of robustness in the face of  model uncertainty in finance.
\end{abstract}

\section{Introduction}

The purpose of this paper is twofold. On the one hand, it deals with Constant Proportion Portfolio Insurance (CPPI) and its dynamic extension, which may be called Dynamic Proportion Portfolio Insurance (DPPI). On the other hand, it deals with the general issues of model uncertainty and model risk in finance by presenting a case study in which a  problem of dynamic trading can be solved in a probability-free manner. 

Constant Proportion Portfolio Insurance (CPPI) was first studied by Perold \cite{Perold}, Black and Jones \cite{BlackJones}, and Black and Perold \cite{BlackPerold}. It provides a strategy that yields superlinear participation in future asset returns while retaining a  security guarantee on a part  of the invested capital (\lq\lq the floor"). 
 In the Black \& Scholes framework, which is the basis for most academic studies on CPPI, constructing a CPPI strategy is equivalent to hedging a certain power option. Moreover, in this framework, the CPPI strategy has no \emph{gap risk} in the sense that its value  stays above the floor with probability one. On the other hand, Cont and Tankov \cite{ContTankovCPPI}, Balder et al. \cite{Balder}, and Paulot and Lacroze \cite{Paulot} show that the CPPI strategy may break through the floor in incomplete market models in which asset prices may jump or in which the portfolio may only be rebalanced at a finite number of trading dates.  In this sense, the CPPI strategy may fail in these settings, and one is faced with the question of quantifying the resulting gap risk, which is important in practice  \cite{Balder,ContTankovCPPI,Paulot}. 

The failure of the CPPI strategy in the incomplete market models of \cite{Balder,ContTankovCPPI,Paulot} on the one hand, and the absence of gap risk in the complete Black \& Scholes framework on the other hand, raise the question  whether the completeness of the underlying market model is  related to the possible nonexistence of gap risk. More generally, one may ask which model features are crucial for setting up a CPPI strategy: 
\begin{itemize}
\item Can one choose every general semimartingale model? 
\item What is the role of arbitrage? 
In particular, must the underlying market model be arbitrage-free to set up the CPPI strategy? 
\item If absence of arbitrage is not essential,  can one even go beyond the class of general semimartingale models and allow for fractal or fractional models such as those in \cite{Benderetal, Cheridito, Salopek, Shiryaev}? 
\item Are there other sources for gap risk apart from jumps in asset prices or discrete rebalancing times? 
\end{itemize}

In this paper, we address all these questions by considering CPPI in the probability-free setting of F\"ollmer's pathwise It\^o calculus \cite{FoellmerIto}; see also \cite{Benderetal1,BickWillinger,DavisRavalObloij,FoellmerECM,FoellmerSchiedBernoulli,SchiedStadje,Sondermann}. In this framework, the dynamics of asset prices are simply described by a single trajectory satisfying a few basic assumptions. In particular, this framework does not postulate any probabilistic mechanism that governs the choice of a particular price evolution. All that is required from the price trajectory of a risky asset is that it is continuous and admits a continuous quadratic variation in a   pathwise sense. These two conditions are satisfied, in particular, by the typical sample paths of any continuous semimartingale, regardless of whether the semimartingale admits an equivalent martingale measure or not. A continuous quadratic variation exists even for a much larger class of trajectories than the class of semimartingale sample paths. An example are the typical sample paths of fractional Brownian motion with Hurst index  $H>\frac12$, which have vanishing quadratic variation.   It is perhaps interesting to note here that vanishing quadratic variation immediately yields the existence of arbitrage opportunities via a simple application of F\"ollmer's pathwise It\^o formula to the function $f(x)=x^2$; see \cite{FoellmerECM} or \cite[Section 5.1]{FoellmerSchiedBernoulli}.

Our first main result will show that, in this very general context, CPPI can be defined as a self-financing trading strategy and that CPPI has no gap risk in the the sense that its value always stays above the floor. This means in particular that neither the completeness nor the absence of arbitrage play any role in the definition of CPPI and for the possible existence of gap risk.  Gap risk is therefore exclusively generated by jumps in the asset price dynamics or by constraints on the rebalancing times of the portfolio.

Our second main result concerns a dynamic extension of the CPPI strategy in which the multiplier level may  depend on quantities including  
time and price evolution. While the possibility of such a Dynamic Proportion Portfolio Insurance (DPPI) has been mentioned several times in the literature, the author was unable to find any corresponding mathematical analysis. Here we treat DPPI in the same strictly pathwise framework as CPPI. We show that in this framework DPPI can always be defined as a self-financing trading strategy and that its value  never breaks through the floor. 

\medskip

The beauty of F\"ollmer's pathwise approach to continuous-time trading lies in the fact that just one single price trajectory is needed. This corresponds to the reality of financial markets, where prices are given only once and the \lq\lq experiment" of pricing a given asset in a specific state of the world can never be repeated.   A proponent of the frequentist interpretation of probability may thus argue that it is therefore anyway impossible to  measure the \lq\lq objective" probability law according to which market scenarios are selected. But even if one does not share such a strong view on the interpretation of probabilistic models of price evolutions,  one will still feel compelled to acknowledge that the complexity of economic dynamics will make it practically impossible to accurately describe the probability law of the price evolution. That is, probabilistic models are subject to Knightian uncertainty and the resulting model risk \cite{Knight}. In recent years, the issue of Knightian uncertainty in finance has received increasing attention; see, e.g., \cite{BrownHobsonRogers,Cont,CoxObloj,FoellmerSchied,GilboaSchmeidler,HansenSargent,Maccheronietal} and \cite[Section 5]{FoellmerSchiedBernoulli}. In this context, the pathwise approach is remarkable as it completely avoids the choice of a probabilistic model. It was known previously that, for example, hedging strategies for variance swaps and related derivatives could be constructed within this pathwise framework \cite{DavisRavalObloij,FoellmerSchiedBernoulli}. The present paper now adds that also CPPI strategies can be constructed in a purely pathwise manner, so that our result can also be viewed as a case study in model uncertainty. 

In the subsequent Section~\ref{Statements Section} we first recall some basic facts about F\"ollmer's pathwise It\^o calculus and its financial implementation. Our main results on  CPPI and DPPI strategies are stated in Theorems~\ref{CPPIProp} and~\ref{main thm}, respectively. The proofs of these results are based on the \emph{associativity} of F\"ollmer's pathwise It\^o integral, which  is a result of independent interest. It is stated, among some other facts on pathwise It\^o calculus, in Section~\ref{Ito Section}. The proofs of Theorems~\ref{CPPIProp} and~\ref{main thm} are given in Section~\ref{Proofs Section}. 

\section{Statement of results}\label{Statements Section}

 Constant Proportion Portfolio Insurance (CPPI) is a self-financing investment strategy that allows for a superlinear participation in future assert returns while simultaneously retaining a guaranteed capital level. 
In the academic literature, this strategy has so far been discussed within various probabilistic models for the evolution of the price process. A common feature of these studies is that price processes are assumed to be semimartingales and market models are often taken as complete. Yet, it is a well-known fact that in a financial context the choice of a probabilistic model is typically itself subject to Knightian uncertainty; see, e.g.,  \cite[Section 5]{FoellmerSchiedBernoulli}. Our goal in this paper is to show that this restriction to probabilistic semimartingale models is unnecessary in the case of CPPI  strategies. We will show that the strategy  works in a strictly pathwise setting that not only includes all continuous semimartingales but also applies to the sample paths of many stochastic processes that are not semimartingales such as price processes based on fractional Brownian motion.  More precisely, we will work in a probability-free framework that is based on  F\"ollmer's pathwise It\^o calculus \cite{FoellmerIto}. In the context of a financial market model,  this pathwise It\^o calculus has been applied to the hedging of derivatives in 
 \cite{BickWillinger} and \cite{FoellmerECM}; see also \cite{Sondermann} for an introduction  and \cite[Section 5.1]{FoellmerSchiedBernoulli} for a short, recent survey.

The beauty of  the probability-free framework is to assume just one price trajectory as given. We assume that this trajectory includes two assets, a locally riskless bond and a risky asset. Bond prices are described by
\begin{equation}\label{bond eq}
B_t=\exp\Big(\int_0^tr_s\,ds\Big),
\end{equation}
where $r:[0,\infty)\to\bR$ is measurable and satisfies $\int_0^t|r_s|\,ds<\infty$ for all $t>0$. Prices of the risky asset are modeled by a single continuous function $S:[0,\infty)\to (0,\infty)$. 

In discrete time, trading is possible at time points $0=t_0<t_1<\cdots$, and we assume that $\lim_nt_n=+\infty$.  The set $\bT=\{t_0,t_1,\dots\}$ is the corresponding \emph{time grid}. 
Continuous-time trading needs to be defined in terms of an approximation from discrete time. To this end, we fix a  sequence $(\bT_N)_{N\in\bN}$ of time grids satisfying $\bT_1\subset\bT_2\subset\cdots$ and $\lim_N\sup_{t_i\in\bT_N}|t_{i+1}-t_i|=0$. An example of such a sequence is provided by the dyadic time grids, $\bT_N=\{k2^{-N}\,|\,k=0,1,\dots\}$. Following F\"ollmer \cite{FoellmerIto}, we will say that a continuous trajectory $X:[0,\infty)\to\bR$ has  \emph{continuous quadratic variation $[X]$ along the sequence $(\bT_N)$} if for each $t>0$ the limit
\begin{equation}\label{quadratic variation}
[X]_t:=\lim_{N\ua\infty} \sum_{\stackrel{t_i, t_{i+1}\in\bT_N}{t_{i+1}\leq t}} 
\left(X_{t_{i+1}}-X_{t_i}\right)^2
\end{equation}
exists, and if $t\mapsto[X]_t$ becomes a continuous function on $[0,\infty)$ for the choice $[X]_0=0$. Note that $t\mapsto[X]_t$ is nondecreasing and hence locally of finite variation. 
The existence of the continuous quadratic variation $[X]$  along $(\bT_N)$ guarantees that $X$ can serve as an integrator in F\"ollmer's pathwise It\^o calculus \cite{FoellmerIto}. 
We  state below the corresponding pathwise It\^o formula in the form in which it will be needed for the statement of our results on CPPI. Their proofs will require a  more general, multidimensional version, which is given in Section~\ref{Ito Section}. 

The class  $C^{1,2}(\bR^n\times\bR)$ will consist of all functions $f(\bm a,x)$ that are continuously differentiable in $(\bm a,x)\in \bR^n\times\bR$ and twice continuously differentiable in $x\in\bR$.  We will write $f_{a^k}$ for the partial derivative of $f$ with respect to the $k^{\text{th}}$ coordinate of the vector $\bm a=(a^1,\dots, a^n)$ and $f_x$ and $f_{xx}$ for the first and second partial derivatives with respect to $x$. 

\medskip

\begin{theorem}[F\"ollmer \cite{FoellmerIto}]\label{FoellmerThm}  Suppose that the continuous trajectory $X$ admits the continuous quadratic variation $[X]$  along $(\bT_N)$, that $\bm A:[0,\infty)\to\bR^n$ is a continuous function whose components are locally of finite variation, and that $f\in C^{1,2}(\bR^n\times\bR)$. Then 
\begin{eqnarray*}f(\bm A_t,X_t)-f(\bm A_0,X_0)=\sum_{k=1}^n\int_0^tf_{a^k}(\bm A_s,X_s)\,dA_s^k+\int_0^tf_x(\bm A_s,X_s)\,dX_s+\frac12\int_0^tf_{xx}(\bm A_s,X_s)\,d[X]_s,
\end{eqnarray*}
where $\int_0^tf_{a^k}(\bm A_s,X_s)\,dA_s^k$ and $\int_0^tf_{xx}(\bm A_s,X_s)\,d[X]_s$ are taken in the usual sense of Riemann--Stieltjes integrals and the \emph{pathwise It\^o integral} $\int_0^tf_x(\bm A_s,X_s)\,dX_s$ is given by the following limit of nonanticipative Riemann sums:
\begin{equation}\label{pathwise Ito integral}
\int_0^tf_x(\bm A_s,X_s)\,dX_s=\lim_{N\ua\infty}\sum_{\stackrel{t_i, t_{i+1}\in\bT_N}{t_{i+1}\leq t}}f_x(\bm A_{t_i},X_{t_i})(X_{t_{i+1}}-X_{t_i}).
\end{equation}
\end{theorem}

\medskip

The preceding theorem implies in particular 
the existence of the  pathwise It\^o integral \eqref{pathwise Ito integral}. We therefore can define a class of admissible integrands:

\medskip

\begin{definition}\label{AdmissibleIntegrandd=1Def} Suppose that the continuous trajectory $X$ admits the continuous quadratic variation $[X]$  along $(\bT_N)$. 
A real-valued function $t\mapsto\xi_t$ is called an \emph{admissible integrand for $X$} if for each $T>0$ there exists $n\in\bN$, a function $g\in C^1(\bR^{n+1})$, and a continuous function $\bm A:[0,\infty)\to\bR^n$  whose components are of finite variation on $[0,T]$ such that $\xi_t=g(\bm A_t,X_t)$ for $0\le t\le T$.
\end{definition}

When $\xi$ is an admissible integrand for $X$ and $g$ and $\bm A$ are as in Definition~\ref{AdmissibleIntegrandd=1Def}, then $f(\bm a,x):=\int_0^x g(\bm a,y)\,dy$ belongs to $C^{1,2}(\bR^n\times\bR)$, and Theorem~\ref{FoellmerThm} implies that the It\^o integral
$$\int_0^tf_x(\bm A_s,X_s)\,dX_s=\int_0^t g(\bm A_s,X_s)\,dX_s=\int_0^t\xi_s\,dX_s
$$ 
can be defined through the limit on the right-hand side of \eqref{pathwise Ito integral}.

\medskip

Let us now return to our financial context, in which bond prices are given by \eqref{bond eq} and prices of the risky asset are modeled by a continuous path $S:[0,\infty)\to (0,\infty)$. We will assume from now on that $S$ admits the continuous quadratic variation $[S]$ along $(\bT_N)$. 
A trading strategy will be a pair $(\xi,\eta)$ of of functions $[0,\infty)$, where $\xi_t$ describes the number of shares in the risky asset that are held at time $t$, while $\eta_t$ stands for the number of shares in the bond. Using pathwise It\^o calculus, 
we can now define the notion of a self-financing trading strategy:

\begin{definition}\label{self-financingDef}Let $(\xi,\eta)$ be a pair of real-valued measurable functions such that $\xi$ is an admissible integrand for $S$ and $\int_0^t|\eta_sr_s|\,ds<\infty$  for all $t\ge0$. The pair $(\xi,\eta)$ is called a \emph{self-financing  strategy} if the corresponding \emph{portfolio value},
$$V_t:=\xi_tS_t+\eta_tB_t,\quad t\ge0,
$$
satisfies the identity
$$V_t=V_0+\int_0^t\xi_s\,dS_s+\int_0^t\eta_s\,dB_s,\quad t\ge0.
$$
\end{definition}

\medskip

\begin{remark}
In the preceding definition, trading strategies are based on the notion of admissible integrands introduced in Definition~\ref{AdmissibleIntegrandd=1Def}. It is worth pointing out that  this notion allows for a large class of integrands, which, for instance, includes the delta hedging strategies for many practically relevant exotic and plain-vanilla options in Markovian market models such as geometric Brownian motion or local volatility; see \cite{SchiedStadje}. Moreover, for $\bm A$ in Definition~\ref{AdmissibleIntegrandd=1Def} one can take a continuous function of moving averages, $t\mapsto \int_{(t-\delta)^+}^t S_s\,ds$, or running maxima, $t\mapsto \max_{(t-\delta)^+\le s\le t} S_s$, because these are continuous functions of $t$ with  finite variation on every interval $[0,T]$. 
\end{remark}

\medskip

We can now proceed toward defining the CPPI strategy in our model-free setting. At time $t=0$, one is given the initial capital $V_0>0$,  a {security level} $\alpha\in[0,1]$, and a multiplier $m>0$.  The security level specifies the proportion of the initial capital that one is not willing to risk. That is, the portfolio value  should never fall below the \emph{floor} $\alpha V_0B_t$, which one would have attained by investing the fraction $\alpha V_0$ of the initial capital into the bond right from the start.   Now suppose 
 that  the portfolio value $V_t$ of the CPPI strategy at time $t$ is already given. The amount
\begin{equation}\label{cushion}
C_t:=V_t-\alpha V_0B_t
\end{equation}
by which the portfolio value exceeds the floor $\alpha V_0B_t$ is called the \emph{cushion}.  
  The cushion should always be nonnegative so that the amount $V_t$ is indeed bounded from below by $\alpha V_0B_t$ at any time.
 In executing a CPPI strategy we invest  a  multiple $m>0$  of the cushion  into the risky asset. That means that we should have
\begin{eqnarray}\label{xi}
\xi_t:=\frac{m C_t}{S_t}.
\end{eqnarray}
The remaining capital is invested into the riskless asset, i.e.,
\begin{eqnarray} \label{eta}
\eta_t=\frac{V_t-\xi_tS_t}{B_t}=\frac{V_t-mC_t}{B_t}.
\end{eqnarray}
Note that we can  have $mC_t>V_t$, which means that it is possible that the CPPI strategy does not include any risk-free investment and, instead, is short in cash. The formulas \eqref{xi} and \eqref{eta} provide a feedback description of the CPPI strategy. It is, however, not clear \emph{a priori} that this feedback description  gives rise to a self-financing strategy. More precisely, the following  question arises:
\begin{itemize}
\item Does there exist a self-financing strategy $(\xi,\eta)$ whose portfolio value  $V_t=\xi_t S_t+\eta_t B_t$ is such that the identities \eqref{cushion}, \eqref{xi}, and \eqref{eta} hold?
\end{itemize}
When this question can be answered affirmatively, the following two questions arise: 
\begin{itemize}
\item  Is the CPPI strategy free of gap risk?  That is, does the portfolio value $V_t$ of the CPPI strategy always exceed the floor $\alpha V_0B_t$ or, equivalently, do we have $C_t\ge0$ for all $t\ge0$?
\item Are CPPI strategies unique in the sense that there can be at most one unique self-financing strategy $(\xi,\eta)$ such that \eqref{cushion}, \eqref{xi}, and \eqref{eta} hold? 
\end{itemize}
Our first main results yields that all three questions can be answered affirmatively. In view of the generality of our setup, this result implies in particular that notions of market completeness or absence of arbitrage are not needed for CPPI to work. It also follows that  gap risk  is not caused by issues such as market incompleteness. Gap risk only results when one is not able to instantaneously adjust the portfolio in response to asset price changes, as it occurs in the presence of jumps \cite{ContTankovCPPI} or under constraints on the available trading dates \cite{Balder,Paulot}.

\medskip

\begin{theorem}\label{CPPIProp} For given $V_0\ge0$, $\alpha\in[0,1]$, and   $m>0$, we define
\begin{align}C_t&=(1-\alpha)V_0\Big(\frac{S_t}{S_0}\Big)^mB_t^{1-m}e^{-\frac{1}2m(m-1)[\log S]_t}\label{CPPIcushionDefEq}
\end{align}
and 
\begin{equation}\label{CPPIVtDefEq}
V_t:=C_t+\alpha V_0 B_t.
\end{equation}
Then the equations \eqref{cushion}, \eqref{xi}, and \eqref{eta} define a self-financing  strategy $(\xi,\eta)$  with associated portfolio value  $V$. In particular the CPPI strategy has no gap risk in the sense that its portfolio value always stays above the floor:
$$V_t\ge \alpha V_0B_t\qquad\text{for all $t\ge0$.}
$$
Moreover, $(\xi,\eta)$ is the unique self-financing trading strategy for which \eqref{cushion}, \eqref{xi}, and \eqref{eta} are satisfied.
\end{theorem}

\medskip

\begin{remark}It is interesting to analyze the various terms in \eqref{CPPIcushionDefEq} in regards of their contributions to the return of the CPPI strategy. In a Black--Scholes setting, which provides the framework for most academic studies on CPPI strategies, $B_t$ and $[\log S]_t $ are deterministic quantities and  can be treated as constants when $t$ is fixed. Therefore the performance of the CPPI strategy can be described as  a constant times the $m^{\text{th}}$ power, $(S_t/S_0)^m$, of asset returns.  This view, however, conceals some of the downside risks that are associated with volatile model parameters. The impact of volatile interest rates is described by the term $B_t^{1-m}$, which, for the common case $m>1$, will decrease returns when interest rates go up. Next, the term $e^{-\frac{1}2m(m-1)[\log S]_t}$ describes the influence of  volatility on the return of the CPPI strategy, because 
\begin{equation}\label{realized variance}
[\log S]_t=\lim_{N\ua\infty} \sum_{\stackrel{t_i, t_{i+1}\in\bT_N}{t_{i+1}\leq t}} 
\left(\log S_{t_{i+1}}-\log S_{t_i}\right)^2
\end{equation}
 is often called the realized variance of $S$.  The reason for this terminology is the fact that an approximating sum on the right-hand side of \eqref{realized variance} can be regarded as the payoff of a variance swap with maturity $t$; see \cite{Buehler}. Moreover, it follows from \cite[Proposition 2.2.10]{Sondermann} that $[\log S]_t=\int_0^t\sigma_s^2\,ds$ when $d[S]_t=\sigma_t^2S_t^2\,dt$.   Formula \eqref{CPPIcushionDefEq} thus states that an increase in realized variance adversely impacts returns by way of the exponential function $x\mapsto e^{-\frac{1}2m(m-1)x}$.
\end{remark}

\medskip

Theorem~\ref{CPPIProp} is in fact a corollary of our following, more general result. It deals with the situation in which the multiplier $m$ is not chosen as a constant but may vary in time.  Such an extension to \emph{Dynamic Proportion Portfolio Insurance (DPPI)} is natural, because  the multiplier $m$ in the CPPI strategy can be regarded as a measure for the leverage of the CPPI investment strategy, and one may wish to choose varying amounts of leverage over time. For instance for a pension fund with a fixed retirement date it can make sense to start off with a high leverage and to revert to a more conservative, lower leverage factor as retirement approaches. Moreover, leverage should be allowed to depend on the current spot, interest rates, and on performance indicators such as realized variance, moving averages, or running maxima.
We model this dynamic adjustment of  leverage  by a continuous   multiplier function $m_t\ge0$. As before, when a security level $\alpha\in[0,1]$ and the value of the investment strategy $V_t$ at time $t$ are given, we define the cushion $C_t$ by 
\begin{equation}\label{cushion DPPI}
C_t=V_t-\alpha V_0B_t\ge0
\end{equation}
 and make the following respective allocations into risky asset and bond:
\begin{equation}\label{xi eta}
\xi_t=\frac{m_tC_t}{S_t}\qquad\text{and}\qquad \eta_t=\frac{V_t-\xi_tS_t}{B_t}=\frac{V_t-m_tC_t}{B_t}.
\end{equation}

\medskip

\begin{theorem}\label{main thm}Suppose that $\alpha\in[0,1]$ and $V_0\ge0$ are given and that $m_t$ is an admissible integrand for $S$. Then $m_t/S_t$ is an admissible integrand for $S$. If we define
\begin{equation}\label{Ct eq}
C_t:=(1-\alpha)V_0\exp\bigg(\int_0^t\frac{m_s}{S_s}\,d S_s-\frac12\int_0^t\frac{m_s^2}{S_s^2}\,d[ S]_s+\int_0^t(1-m_s)r_s\,ds\bigg)
\end{equation}
and
\begin{equation}\label{}
V_t:=C_t+\alpha V_0B_t,
\end{equation}
then \eqref{cushion DPPI} and \eqref{xi eta} defines a self-financing trading strategy with portfolio value  $V$. In particular, the DPPI strategy has no gap risk in the sense that its value never breaks through the floor:
$$V_t\ge \alpha V_0B_t\qquad\text{for all $t\ge0$.}
$$
Moreover, $(\xi,\eta)$ is the unique self-financing trading strategy for which \eqref{cushion DPPI} and \eqref{xi eta}  are satisfied.

\end{theorem}

\medskip


\section{Complements on F\"ollmer's pathwise It\^o calculus}\label{Ito Section}

Pathwise It\^o calculus goes back to F\"ollmer \cite{FoellmerIto}, where a strictly pathwise It\^o formula was proved. This topic was further developed in the lectures of Hans F\"ollmer, some of which form the basis of the book \cite{Sondermann}. The proofs of   Theorems~\ref{CPPIProp} and~\ref{main thm}   require some techniques in pathwise It\^o calculus that go beyond the material in \cite{FoellmerIto,Sondermann}. In particular, we need the so-called associativity of the pathwise It\^o integral. This property is stated in Theorem~\ref{associativity thm} and is of independent interest.

The statements of Theorems~\ref{CPPIProp} and~\ref{main thm} involve only the pathwise It\^o formula in the one-dimensional form of Theorem~\ref{FoellmerThm}; their proofs and Theorem~\ref{associativity thm}  require a $d$-dimensional integrator $\bm X_t=(X^1_t,\dots,X^d_t)$. So let us recall the  pathwise It\^o formula in the multidimensional form in which it will henceforth be needed. To enhance the readability, we will write multidimensional objects in boldface type. 

We  fix a  sequence $(\bT_N)_{N\in\bN}$ of time grids satisfying $\bT_1\subset\bT_2\subset\cdots$ and $\lim_N\sup_{t_i\in\bT_N}|t_{i+1}-t_i|=0$. We moreover suppose that  $\bm X:[0,\infty)\to\bR^d$ is continuous and that for all $k$ and $m$ the real-valued path $X^k_t+X^m_t$ has continuous quadratic variation $[X^k+X^m]$. This assumption is equivalent to the existence of the \emph{covariation of $X^k$ and $X^m$} defined by
\begin{equation}\label{covariation eq}
\begin{split}
[X^k,X^m]_t&:=\frac12\Big([X^k+X^m]_t-[X^k]_t-[X^m]_t\Big)\\
&=\lim_{N\ua\infty} \sum_{\stackrel{t_i, t_{i+1}\in\bT_N}{t_{i+1}\leq t}} 
\big(X^k_{t_{i+1}}-X^k_{t_i}\big)\big(X^m_{t_{i+1}}-X^m_{t_i}\big).
\end{split}
\end{equation}
Here the latter identity follows by polarization of the corresponding sums in \eqref{quadratic variation}. Clearly, $[X^k]$ exists as $\frac14[X^k+X^k]$. Note that $[X^k,X^m]_t$ is locally of finite variation as a function of $t$, because it is the difference of the nonincreasing functions $[X^k+X^m]_t$ and $[X^k]_t+[X^m]_t$. 

\medskip

\begin{remark}\label{FV remark}We will need the following facts that   easily follow from Propositions 2.2.2, 2.2.9, and 2.3.2 in \cite{Sondermann}. Suppose that $Y$ is continuous and admits the continuous quadratic variation $[Y]$ along $(\bT_N)$ and $A$ is continuous and locally of finite variation. Then both $[A]$ and $[Y+A]$ exist  along $(\bT_N)$  and are given by $[A]=0$ and $[Y+A]=[Y]$. By means of the polarization identity \eqref{covariation eq} we get moreover that $[Y,A]=0$.
\end{remark}

\medskip

The class  $C^{1,2}(\bR^n\times\bR^d)$ will consist of all functions $f(\bm a,\bm x)$ that are continuously differentiable in $(\bm a,\bm x)\in \bR^n\times\bR^d$ and twice continuously differentiable in $\bm x\in\bR^d$.  We will write $f_{a^k}$ for the partial derivative of $f$ with respect to the $k^{\text{th}}$ component of $\bm a=(a^1,\dots, a^n)$. The gradient of $f$ in direction $\bm x=(x^1,\dots, x^d)$ will be denoted by 
$$\nabla_{\bm x}f=\Big(\frac{\partial f}{\partial x^1},\dots, \frac{\partial f}{\partial x^d}\Big),$$
 and we will write $f_{x^kx^m}$ for the  second partial derivatives with respect to the components $x^k$ and $x^m$ of the vector $\bm x$. The Euclidean inner product of  two vectors $\bm x$ and $\bm y$ will be denoted by $\bm x\cdot\bm y$.

\medskip 

\begin{theorem}[F\"ollmer \cite{FoellmerIto}]\label{FoellmerThm2}  Suppose that the continuous trajectory $\bm X:[0,\infty)\to\bR^d$ admits  for all $k$ and $m$ the continuous covariation $[X^k,X^m]$  along $(\bT_N)$, that $\bm A:[0,\infty)\to\bR^n$ is a continuous function whose components are locally of finite variation, and that $f\in C^{1,2}(\bR^n\times\bR^d)$. Then 
\begin{eqnarray*}f(\bm A_t,\bm X_t)-f(\bm A_0,\bm X_0)&=&\int_0^t\nabla_{\bm x}f(\bm A_s,\bm X_s)\,d\bm X_s+\sum_{k=1}^n\int_0^tf_{a^k}(\bm A_s,\bm X_s)\,dA_s^k\\&&+\ \frac12\sum_{k,m=1}^d\int_0^tf_{x^kx^m}(\bm A_s,\bm X_s)\,d[X^k,X^m]_s,
\end{eqnarray*}
where $\int_0^tf_{a^k}(\bm A_s,\bm X_s)\,dA_s^k$ and $\int_0^tf_{x^kx^m}(\bm A_s,\bm X_s)\,d[X^k,X^m]_s$ are taken in the usual sense of  Stieltjes integrals and the \emph{pathwise It\^o integral} is given by the following limit of nonanticipative Riemann sums:
\begin{equation}\label{pathwise Ito integral d}
\int_0^t\nabla_{\bm x}f(\bm A_s,\bm X_s)\,d\bm X_s=\lim_{N\ua\infty}\sum_{\stackrel{t_i, t_{i+1}\in\bT_N}{t_{i+1}\leq t}}\nabla_{\bm x}f(\bm A_{t_i},\bm X_{t_i})\cdot(\bm X_{t_{i+1}}-\bm X_{t_i}).
\end{equation}
\end{theorem}

\begin{proof}For $f\in C^2(\bR^{n+d})$ the result follows from Remarque 1 in \cite{FoellmerIto} and by noting that the quadratic variations $[A^k]$ and covariations $[A^k,A^\ell]$ and $[A^k,X^i]$ ($k,\ell=1,\dots,n$, $i=1\dots, d$) vanish identically according to Remark~\ref{FV remark}. The extension to $f\in C^{1,2}(\bR^n\times\bR^d)$ is obtained just as in the proof of our Theorem~\ref{associativity thm} below by using Taylor development of $f(\bm a,\bm x)$ up to first order in $\bm a$ and up to second order in $\bm x$. 
\end{proof}

\begin{remark}In \eqref{pathwise Ito integral d} it is typically \emph{not} possible to write 
$$\int_0^t\nabla_{\bm x}f(\bm A_s,\bm X_s)\,d\bm X_s=\sum_{i=1}^d\int_0^tf_{x^i}(\bm A_s,\bm X_s)\,d X^i_s,
$$
because the  integrals $\int_0^tf_{x^i}(\bm A_s,\bm X_s)\,d X^i_s$ on the right-hand side need not exist individually as the limits of nonanticipative Riemann sums. 
\end{remark}

\medskip

Theorem~\ref{FoellmerThm2} implies in particular that the pathwise It\^o integral $\int_0^t\bm\xi_s\,d\bm X_s$ can be defined via 
\eqref{pathwise Ito integral d} when the integrand $\bm \xi$ is of the form $\bm \xi_t=\nabla_{\bm x}f(\bm A_s,\bm X_s)$ for some continuous function $\bm A:[0,\infty)\to\bR^n$  whose components are locally of finite variation and for $f\in C^{1,2}(\bR^n\times\bR^d)$. Since in the case $d>1$ not every $C^1$-function $\bm g:\bR^n\times\bR^d\to\bR^d$ is of the form $\bm g=\nabla_{\bm x}f$ for some $f\in C^{1,2}(\bR^n\times\bR^d)$, the following definition of $d$-dimensional admissible integrands needs to be slightly more complicated than its one-dimensional counterpart,
Definition~\ref{AdmissibleIntegrandd=1Def}.

\medskip

\begin{definition}Suppose that the continuous trajectory $\bm X:[0,\infty)\to\bR^d$ admits the continuous covariations $[X^k,X^m]$ along $(\bT_N)$, $k,m=1,\dots,d$.  
A function $t\mapsto\bm\xi_t\in\bR^d$ is called an \emph{admissible integrand for $\bm X$} if for each $T>0$ there exists $n\in\bN$, a function $f\in C^{1,2}(\bR^{n}\times\bR^d)$, and a continuous function $\bm A:[0,\infty)\to\bR^n$  whose components are of finite variation on $[0,T]$ such that $\bm\xi_t=\nabla_{\bm x}f(\bm A_t,X_t)$ for $0\le t\le T$.
\end{definition}

\medskip

The following result is a straightforward extension of \cite[Proposition 2.3.3]{Sondermann}, and its proof is left to the reader.

\medskip

\begin{proposition}\label{covariation Prop}Suppose that  $\bm X$ is as in Theorem~\ref{FoellmerThm2}, that  $\bm\xi^{(1)},\dots,\bm\xi^{(\nu)}$ are admissible integrands for $\bm X$, and that
$$\displaystyle Y^\ell_t:=\int_0^t\bm\xi^{(\ell)}_s\,d\bm X_s,\qquad \ell=1,\dots,\nu.$$ Then $\bm Y_t=(Y^1_t,\dots,Y^\nu_t)$  is a continuous trajectory that admits the continuous covariations
$$[Y^k,Y^\ell]_t=\sum_{i,j=1}^d\int_0^t\xi^{(k),i}_s\xi^{(\ell),j}\,d[X^i,X^j]_s,\qquad k,\ell=1,\dots,\nu.
$$
\end{proposition}

\medskip

The preceding proposition implies in particular that $\bm Y_t=(Y^1_t,\dots,Y^\nu_t)$ is again an admissible integrator for pathwise It\^o calculus. The following \emph{associativity rule} for the pathwise It\^o integral  shows that one can express a pathwise It\^o integral with respect to $\bm Y$ as a pathwise It\^o integral with respect to $\bm X$. 

\medskip

\begin{theorem}[Associativity of the pathwise It\^o integral]\label{associativity thm}Suppose that  $\bm X$,  $\bm\xi^{(1)},\dots,\bm\xi^{(\nu)}$ and $\bm Y$ are as in Proposition~\ref{covariation Prop}, and let $\bm\eta=(\eta^1,\dots,\eta^\nu)$ be an admissible integrand for $\bm Y$. Then $\sum_{\ell=1}^\nu\eta^\ell\bm\xi^{(\ell)}$ is an admissible integrand for $\bm X$ and
$$\int_0^t\bm\eta_s\,d\bm Y_s=\int_0^t\sum_{\ell=1}^\nu\eta_s^\ell\bm\xi^{(\ell)}_s\,d\bm X_s.
$$
\end{theorem}

\begin{proof}[Proof of Theorem~\ref{associativity thm}] We fix $T\ge0$. For $t\le T$,  let $\bm\xi^{(\ell)}$ be of the form $\bm\xi^{(\ell)}_t=\nabla_{\bm x}f^\ell(\bm A^{(\ell)}_t,\bm X_t)$ for $n_\ell \in\bN$, continuous $\bm A^{(\ell)}:[0,T]\to\bR^{n_\ell }$ with components of finite variation,  and $f^\ell\in C^{1,2}(\bR^{n_\ell }\times\bR^d)$. We also define
\begin{equation}\label{Ain+1ell}
A^{(\ell),n_\ell +1}_t:=\sum_{k=1}^{n_\ell }\int_0^tf^\ell _{a^k}(\bm A^{(\ell)}_s,\bm X_s)\,dA_s^k+ \frac12\sum_{k,m=1}^d\int_0^tf^\ell _{x^kx^m}(\bm A^{(\ell)}_s,\bm X_s)\,d[X^k,X^m]_s.
\end{equation}
Then $A^{(\ell),n_\ell +1}$ is continuous and of finite variation on $[0,T]$ by standard properties of Stieltjes integrals (see   \cite[Theorem I.5c]{Widder}).
Moreover, the pathwise It\^o formula from Theorem~\ref{FoellmerThm} implies that 
\begin{equation}\label{Y wt f eq}
Y^\ell _t=f^\ell (\bm A^{(\ell)}_t,\bm X_t)-f^\ell(\bm A^{(\ell)}_0,\bm X_0)-A^{(\ell),n_\ell +1}_t=F^\ell (\wt {\bm A}^{(\ell)}_t,\bm X_t)
\end{equation}
where
$$\wt{\bm  A}^{(\ell)}_t:=(A^{(\ell),1}_t,\dots,A^{(\ell),n_\ell }_t,A^{(\ell),n_\ell +1}_t)$$
 and 
$$F^\ell (\wt{\bm a},\bm x):=f^\ell (\bm a,\bm x)-f^\ell (\bm A_0,\bm X_0)-a^{n_\ell +1}\qquad\text{for $\wt {\bm a}=(\bm a,a^{n_\ell +1})\in\bR^{n_\ell }\times\bR$.}$$
  Clearly, $\wt {\bm A}^{(\ell)}:[0,T]\to\bR^{n_\ell +1}$ is continuous and has  finite total variation on $[0,T]$, and $F^\ell $ belongs to $C^{1,2}(\bR^{n_\ell +1}\times\bR^d)$. Moreover,
\begin{equation}\label{nabla Fell id}
\nabla_{\bm x}F^\ell( \wt{\bm a},\bm x) =\nabla_{\bm x}f^\ell(\bm a,\bm x)\qquad\text{for $\wt {\bm a}=(\bm a,a^{n_\ell +1})\in\bR^{n_\ell }\times\bR$.}
\end{equation}
Let us denote 
$${\bm F}(\bm a,\bm x):=\big(F^{1}(\wt{\bm a}^{({1})},\bm x),\dots,F^{\nu}(\wt{\bm a}^{(\nu)},\bm x)\big)\qquad\text{for }\bm a=(\wt{\bm a}^{(1)},\dots,\wt{\bm a}^{(\nu)})\in\bR^{n_1+\cdots+n_\nu+\nu}.$$
By writing $ {\bm A}_t:=(\wt{\bm A}^{(1)},\dots, \wt{\bm A}^{(\nu)})$, the identity \eqref{Y wt f eq} becomes
\begin{equation}\label{Y wt f eq 2} 
\bm Y_t={\bm F}( {\bm A}_t,\bm X_t).
\end{equation}

Since $\bm\eta$ is an admissible integrand for $\bm Y$, there are $m\in\bN$, $h\in C^{1,2}(\bR^{m}\times\bR^\nu)$, and continuous $\bm D:[0,T]\to\bR^m$ with  finite variation such that $\bm\eta_t=\nabla_{\bm y}h(\bm D_t,\bm Y_t)$ for $0\le t\le T$. 
Using \eqref{Y wt f eq 2},  \eqref{nabla Fell id}, and the notation $\nabla_{\bm x}{\bm F}( {\bm a},\bm x)$ for the Jacobi matrix of $\bm x\mapsto {\bm F}( {\bm a},\bm x)$, we get
\begin{eqnarray*}\sum_{\ell=1}^\nu\eta_t^\ell\bm\xi_t^{(\ell)}&=&\sum_{\ell=1}^\nu h_{y^\ell}(\bm D_t,\bm Y_t)\nabla_{\bm x}f^{{\ell}}(\bm A^{(\ell)}_t,\bm X_t)=\nabla_{\bm y}h(\bm D_t,{\bm F}( {\bm A}_t,\bm X_t))\cdot\nabla_{\bm x}{\bm F}( {\bm A}_t,\bm X_t)\\
&=&\nabla_{\bm x}\wt h(\wt{\bm D}_t,\bm X_t),
\end{eqnarray*}
where $\wt{\bm  D}_t=(\bm D_t, {\bm  A}_t)$, and $\wt h((\bm d, {\bm a}),\bm x):=h(\bm d,\bm F( {\bm a},\bm x))$ belongs to $C^{1,2}(\bR^{K}\times\bR^d)$ for $K=m+n_1+\cdots+n_\nu+\nu$.  It follows in particular that  $\sum_{i=1}^\nu\eta^\ell\bm\xi^{(\ell)}$ is an admissible integrand for $\bm X$. 

 The definition \eqref{pathwise Ito integral d} of the It\^o integral and \eqref{Y wt f eq 2} imply that 
\begin{equation}\label{associativity limit}
\begin{split}
\int_0^t\bm\eta_s\,d\bm Y_s&=\lim_{N\ua\infty}\sum_{\stackrel{t_i, t_{i+1}\in\bT_N}{t_{i+1}\leq t}}\bm\eta_{t_i}\cdot\left({\bm F}({\bm A}_{t_{i+1}},\bm X_{t_{i+1}})-{\bm F}({\bm A}_{t_i},\bm X_{t_i})\right)\\
&=\lim_{N\ua\infty}\sum_{\stackrel{t_i, t_{i+1}\in\bT_N}{t_{i+1}\leq t}}\sum_{\ell=1}^\nu\eta^\ell_{t_i}\left( F^\ell({\bm A}_{t_{i+1}},\bm X_{t_{i+1}})-F^\ell({\bm A}_{t_i},\bm X_{t_i})\right),
\end{split}
\end{equation}
where, by abuse of notation, we write $F^\ell({\bm A}_{t_i},\bm X_{t_i})$ instead of  $F^\ell(\wt{\bm A}^{(\ell)}_{t_i},\bm X_{t_i})$.
 Using  multidimensional  Taylor development up to first order in ${\bm a}$ and up to second order in $\bm x$, we get
\begin{eqnarray*}
\lefteqn{ F^\ell({\bm A}_{t_{i+1}},\bm X_{t_{i+1}})-F^\ell({\bm A}_{t_i},\bm X_{t_i})}\\
&=&F^\ell({\bm A}_{t_{i+1}},\bm X_{t_{i+1}})-F^\ell({\bm A}_{t_{i}},\bm X_{t_{i+1}})+F^\ell({\bm A}_{t_{i}},\bm X_{t_{i+1}})-F^\ell({\bm A}_{t_{i}},\bm X_{t_{i}})\\
&=&\nabla_{ \bm a}F^\ell({\bm A}_{t_i},\bm X_{t_{i}})\cdot ({\bm A}_{t_{i+1}}-{\bm A}_{t_i})+\bm \delta_i^\ell\cdot ({\bm A}_{t_{i+1}}-{\bm A}_{t_i})\\
&&+\nabla_{\bm x}F^\ell({\bm A}_{t_i},\bm X_{t_i})\cdot(\bm X_{t_{i+1}}-\bm X_{t_i})\\
&&+\frac12(\bm X_{t_{i+1}}-\bm X_{t_i})\cdot \nabla_{\bm x}^2F^\ell({\bm A}_{t_i},\bm X_{t_i})(\bm X_{t_{i+1}}-\bm X_{t_i})+(\bm X_{t_{i+1}}-\bm X_{t_i})\cdot\bm\eps_{i}^{\ell}(\bm X_{t_{i+1}}-\bm X_{t_i}),
\end{eqnarray*}
where 
$$\bm\delta^\ell_i=\int_0^1\nabla_{\bm a}F^\ell\big({\bm A}_{t_i}+s({\bm A}_{t_{i+1}}-{\bm A}_{t_{i}}),\bm X_{t_{i+1}}\big)\,ds-\nabla_{\bm a}F^\ell({\bm A}_{t_i},\bm X_{t_{i}}),
$$
$\nabla_{\bm x}^2F^\ell$ is the Hessian of $F^\ell$ with respect to $\bm x$,
and, for some $\theta\in[0,1]$,
$$\bm\eps_i^\ell=\frac12\Big(\nabla_{\bm x}^2F^\ell\big({\bm A}_{t_i},\bm X_{t_i}+\theta(\bm X_{t_{i+1}}-\bm X_{t_i})\big)-\nabla_{\bm x}^2F^\ell({\bm A}_{t_i},\bm X_{t_i})\Big).
$$
The continuity of $\nabla_{ \bm a}F^\ell$ and $\nabla_{\bm x}^2F^\ell$ implies that 
$$\max_{\stackrel{t_i, t_{i+1}\in\bT_N}{t_{i+1}\leq T}}\big(|\bm\delta^\ell_i|+\|\bm\eps^\ell_i\|\big)\longrightarrow0\qquad\text{as $N\ua\infty$,}
$$
where $|\cdot|$ denotes the Euclidean norm and $\|\bm\eps^\ell_i\|^2:=\max_{|\bm x|=1}\bm x\cdot \bm\eps^\ell_i\bm x$. When denoting the total variation of $A^k$ over the interval $[0,T]$ by $\|A^k\|_{\text{var}}$, we thus get
\begin{eqnarray*}\Big|\sum_{\stackrel{t_i, t_{i+1}\in\bT_N}{t_{i+1}\leq t}}\bm \delta_i^\ell\cdot ({\bm A}_{t_{i+1}}-{\bm A}_{t_i})\Big|\le \max_{\stackrel{t_i, t_{i+1}\in\bT_N}{t_{i+1}\leq T}}|\bm\delta^\ell_i|\, \sum_{k=1}^K\|A^k\|_{\text{var}}\longrightarrow 0,
\end{eqnarray*}
as $N\ua\infty$. Furthermore, 
\begin{eqnarray*}
\Big|\sum_{\stackrel{t_i, t_{i+1}\in\bT_N}{t_{i+1}\leq t}}(\bm X_{t_{i+1}}-\bm X_{t_i})\cdot\bm\eps_{i}^{\ell}(\bm X_{t_{i+1}}-\bm X_{t_i})\Big|\le \max_{\stackrel{t_i, t_{i+1}\in\bT_N}{t_{i+1}\leq T}}\|\bm\eps^\ell_i\|^2\sum_{\stackrel{t_i, t_{i+1}\in\bT_N}{t_{i+1}\leq t}}(\bm X_{t_{i+1}}-\bm X_{t_i})\cdot(\bm X_{t_{i+1}}-\bm X_{t_i}).
\end{eqnarray*}
Since the rightmost sum converges to the finite limit $[X^1]_t+\cdots+[X^d]_t$,  the right-hand side above tends to zero  as $N\ua\infty$.

Next, the standard existence result for Stieltjes integrals (e.g.,   \cite[Theorem I.4a]{Widder}) implies that 
\begin{eqnarray*}
\lefteqn{\lim_{N\ua\infty}\sum_{\stackrel{t_i, t_{i+1}\in\bT_N}{t_{i+1}\leq t}}\eta^\ell_{t_i}\nabla_{ \bm a}F^\ell({\bm A}_{t_i},\bm X_{t_{i}})\cdot ({\bm A}_{t_{i+1}}-{\bm A}_{t_i})}\\
&=&\sum_{k=1}^K\int_0^t\eta^\ell_sF^\ell_{a^k}({\bm A}_s,\bm X_s)\,dA^k_s\\
&=&\sum_{k=1}^{n_\ell}\int_0^t\eta^\ell_sf^\ell_{a^k}(\bm A^{(\ell)}_s,\bm X_s)\,dA_s^{(\ell),k}-\int_0^t\eta^\ell_s\,dA^{(\ell),n_\ell+1}_s\\
&=&-\frac12\sum_{k,m=1}^d\int_0^tf^\ell_{x^kx^m}(\bm A^{(\ell)}_s,\bm X_s)\,d[X^k,X^m]_s,
\end{eqnarray*}
where we have used \eqref{Ain+1ell} and the  associativity of the Stieltjes integral \cite[Theorem I.6b]{Widder} in the final step.

Next, as observed in \cite{FoellmerIto}, taking $X=X^k$ in \eqref{quadratic variation}, the convergence in \eqref{quadratic variation} can be interpreted as vague convergence of the point measures 
$$\sum_{{t_i, t_{i+1}\in\bT_N}}( X^k_{t_{i+1}}- X^k_{t_i})^2\delta_{t_i}
$$
toward the continuous and nonnegative Radon measure $d[X^k]_t$. Therefore,
$$\sum_{\stackrel{t_i, t_{i+1}\in\bT_N}{t_{i+1}\leq t}}\varphi(t_i)( X^k_{t_{i+1}}- X^k_{t_i})^2\longrightarrow\int_0^t\varphi(s)\,d[X^k]_s
$$
holds for any continuous function $\varphi$ due to the portmanteau theorem (e.g., \cite[Theorem 14.3]{AliprantisBorder}). Via the polarization identity in \eqref{covariation eq}, we get the analogous result for the covariation $[X^k,X^m]$ replacing $[X^k]$. 
This implies 
\begin{eqnarray*}
\lefteqn{\sum_{\stackrel{t_i, t_{i+1}\in\bT_N}{t_{i+1}\leq t}}\eta^\ell_{t_i}(\bm X_{t_{i+1}}-\bm X_{t_i})\cdot \nabla_{\bm x}^2F^\ell({\bm A}_{t_i},\bm X_{t_i})(\bm X_{t_{i+1}}-\bm X_{t_i})}\\
&=&\sum_{k,m=1}^d\sum_{\stackrel{t_i, t_{i+1}\in\bT_N}{t_{i+1}\leq t}}\eta^\ell_{t_i}F^\ell_{x^kx^m}({\bm A}_{t_i},\bm X_{t_i})( X^k_{t_{i+1}}- X^k_{t_i})( X^m_{t_{i+1}}- X^m_{t_i})\\
&\longrightarrow& \sum_{k,m=1}^d\int_0^t\eta^\ell _sF^\ell_{x^kx^m}({\bm A}_{t_i},\bm X_{t_i})\,d[X^k,X^m]_s\\
&=& \sum_{k,m=1}^d\int_0^t\eta^\ell _sf^\ell_{x^kx^m}({\bm A}^{(\ell)}_{t_i},\bm X_{t_i})\,d[X^k,X^m]_s.
\end{eqnarray*}
Moreover, $\nabla_{\bm x}F^\ell({\bm A}_{t_i},\bm X_{t_i})=\nabla_{\bm x}f^\ell({\bm A}^{(\ell)}_{t_i},\bm X_{t_i})=\bm\xi^{(\ell)}_t$, and so 
\begin{eqnarray*}
\sum_{\ell=1}^\nu\sum_{\stackrel{t_i, t_{i+1}\in\bT_N}{t_{i+1}\leq t}}\eta^\ell_{t_i}\nabla_{\bm x}F^\ell({\bm A}_{t_i},\bm X_{t_i})\cdot(\bm X_{t_{i+1}}-\bm X_{t_i})&=&\sum_{\stackrel{t_i, t_{i+1}\in\bT_N}{t_{i+1}\leq t}}\sum_{\ell=1}^\nu\eta^\ell_{t_i}\bm\xi^{(\ell)}_{t_i}\cdot(\bm X_{t_{i+1}}-\bm X_{t_i})\\
&\longrightarrow&\int_0^t\sum_{\ell=1}^\nu\eta_s^\ell\bm\xi^{(\ell)}_s\,d\bm X_s.
\end{eqnarray*}

Putting everything together, we see that the limit on the right-hand side of \eqref{associativity limit} is given by 
\begin{eqnarray*}
\lefteqn{-\frac12\sum_{k,m=1}^d\int_0^tf^\ell_{x^kx^m}(\bm A^{(\ell)}_s,\bm X_s)\,d[X^k,X^m]_s+\frac12\sum_{k,m=1}^d\int_0^tf^\ell_{x^kx^m}(\bm A^{(\ell)}_s,\bm X_s)\,d[X^k,X^m]_s}\\
&&\qquad\qquad\qquad\qquad\qquad\qquad\qquad\qquad\qquad\qquad+\int_0^t\sum_{\ell=1}^\nu\eta_s^\ell\bm\xi^{(\ell)}_s\,d\bm X_s=\int_0^t\sum_{\ell=1}^\nu\eta_s^\ell\bm\xi^{(\ell)}_s\,d\bm X_s.
\end{eqnarray*}
$$
$$
This concludes the proof. \end{proof}

\section{Proofs of Theorems~\ref{main thm} and~\ref{CPPIProp}}\label{Proofs Section}

\begin{proof}[Proof of Theorem~\ref{main thm}] We note first that $m_t/S_t$ is an admissible integrand for $S$, because $m_t$ is an admissible integrand, and $1/S_t$ can locally for $t\in[0,T]$ be written as $f(S_t)$ for some function $f\in C^1(\bR)$ since $S_t$ is bounded away from zero for $0\le t\le T$. In particular,  formula  \eqref{Ct eq} is well-defined.  
Let us write 
$$C_t=X_tA_t,
$$
where
$$X_t:=\exp\bigg(\int_0^t\frac{m_s}{S_s}\,d S_s-\frac12\int_0^t\frac{m_s^2}{S_s^2}\,d[ S]_s\bigg)
$$
and
$$ A_t:=(1-\alpha)V_0\exp\bigg(\int_0^t(1-m_s)r_s\,ds\bigg).
$$
Using the function $f(a,x):=ax$ in Theorem~\ref{FoellmerThm} yields the integration by parts formula
\begin{equation}\label{C int part eq}
C_t-C_0=X_tA_t-X_0A_0=\int_0^tA_s\,dX_s+\int_0^tX_s\,dA_s.
\end{equation}

We now define 
$$Y_t:= \int_0^t\frac{m_s}{S_s}\,d S_s\qquad\text{and}\qquad  L_t:=\frac12\int_0^t\frac{m_s^2}{S_s^2}\,d[ S]_s.
$$
Then $L$ is continuous and of locally finite variation by standard properties of the Stieltjes integral (see   \cite[Theorem I.5c]{Widder}), and $Y$ has the continuous quadratic variation $[Y]_t=2L_t$ by Proposition~\ref{covariation Prop}. 
Moreover, applying Theorem~\ref{FoellmerThm} to the function $g(a,y)=e^{y-a}$ yields 
\begin{eqnarray*}
X_t-X_0&=&g(L_t,Y_t)-g(L_0,Y_0)\\
&=&\int_{0}^t g_y(L_s,Y_s)\,dY_s+\int_0^tg_a(L_s,Y_s)\,dL_s+\frac12\int_0^tg_{yy}(L_s,Y_s)\,d[Y]_s\\
&=&\int_0^tX_s\,dY_s,
\end{eqnarray*}
where we have applied Theorem~\ref{associativity thm} to the Stieltjes integral $\int_0^tg_\ell(L_s,Y_s)\,dL_s$ (instead of Theorem~\ref{associativity thm}   one can here also apply \cite[Theorem I.6b]{Widder}).
We also have
$$A_t-A_0=\int_0^tA_s(1-m_s)r_s\,ds=\int_0^t\frac{A_s(1-m_s)}{B_s}\,dB_s.
$$
Plugging these results into \eqref{C int part eq} and applying Theorem~\ref{associativity thm}  several times yields that 
$$\frac{X_sA_sm_s}{S_s}=\frac{m_sC_s}{S_s}=\xi_s$$
is an admissible integrand for $S$ and that 
\begin{eqnarray}
C_t-C_0&=& \int_0^tA_sX_s\,dY_s+\int_0^t\frac{X_sA_s(1-m_s)}{B_s}\,dB_s\nonumber\\
&=&\int_0^t\frac{m_sC_s}{S_s}\,dS_s+\int_0^t\frac{C_s(1-m_s)}{B_s}\,dB_s.\label{Ct int eqn}
\end{eqnarray}
It follows that $V_t:=C_t+\alpha V_0 B_t$ satisfies
\begin{eqnarray*}
V_t-V_0=\int_0^t\frac{m_sC_s}{S_s}\,dS_s+\int_0^t\frac{C_s(1-m_s)+\alpha V_0B_s}{B_s}\,dB_s=\int_0^t\xi_s\,dS_s+\int_0^t\eta_s\,dB_s,
\end{eqnarray*}
where $\xi$ and $\eta$ are as in \eqref{xi eta}. Finally, we clearly have 
$$V_t=\xi_tS_t+\eta_tB_t,
$$
which shows that $(\xi,\eta)$ is indeed a self-financing strategy with portfolio value  $V$.

\bigskip

Now we turn toward the proof of the uniqueness of the DPPI strategy. To this end, let $(\xi,\eta)$ be the self-financing strategy constructed above, with portfolio value $V_t=\xi_tS_t+\eta_tB_t$ and cushion $C_t=V_t-\alpha V_0B_t$. 
Suppose moreover that $(\wt\xi,\wt\eta)$ is another self-financing strategy with portfolio value $\wt V_t=\wt\xi_tS_t+\wt\eta_tB_t$ and cushion $\wt C_t=\wt V_t-\alpha V_0B_t$ such that $\wt V_0=V_0$ and $\wt\xi_t=m_t\wt C_t/S_t$.

From the self-financing condition we  necessarily have that 
$$\wt\eta_t=\frac{\wt V_t-\wt\xi_tS_t}{B_t}=\frac{\wt C_t+\alpha V_0B_t-m_t\wt C_t}{B_t}$$
and hence
\begin{eqnarray*}
\wt C_t-\wt C_0&=&\wt V_t-V_0-\alpha V_0(B_t-B_0)\\
&=&\int_0^t\frac{m_s\wt C_s}{S_s}\,dS_s+\int_0^t\frac{\wt C_s+\alpha V_0B_s-m_s\wt C_s}{B_s}\,dB_s-\int_0^t\alpha V_0\,dB_s\\
&=&\int_0^t\frac{m_s\wt C_s}{S_s}\,dS_s+\int_0^t(1-m_s)\wt C_s r_s\,ds
\end{eqnarray*}
In the preceding part of the proof we showed that $C_t$ satisfies the same It\^o integral equation; see \eqref{Ct int eqn}. 
When letting
$$Y^{(1)}_t:=e^{-\int_0^t(1-m_s)r_s\,ds}\wt C_t\qquad\text{and}\qquad Y^{(2)}_t:=e^{-\int_0^t(1-m_s)r_s\,ds}C_t,
$$
  one easily checks via \eqref{C int part eq} that $Y^{(i)}$ satisfies
\begin{equation}\label{Yi eq}
Y^{(i)}_t=Y^{(i)}_0+\int_0^t\frac{m_sY^{(i)}_s}{S_s}\,dS_s,\qquad 0\le t\le T,\ i=1,2.
\end{equation}
Here, the pathwise It\^o integral exists since, e.g.,  ${m_tY^{(1)}_t}/{S_t}=e^{-\int_0^t(1-m_s)r_s\,ds}\wt \xi_t$ is clearly an admissible integrand for $S$.

It follows from equation \eqref{Yi eq}, Remark~\ref{FV remark}, and Proposition~\ref{covariation Prop} that the quadratic variations $[Y^{(i)}]$ and  the covariation $[Y^{(1)},Y^{(2)}]$ exist and are given by
\begin{eqnarray}\label{Yi q var eq}
[Y^{(i)}]_t=\int_0^t\frac{m_t^2(Y^{(i)}_t)^2}{S_t^2}\,d[S]_t \qquad\text{and}\qquad [Y^{(1)},Y^{(2)}]_t=\int_0^t\frac{m_t^2Y^{(1)}_tY^{(2)}_t}{S_t^2}\,d[S]_t.
\end{eqnarray}
In particular, $\bm Y_t=(Y^{(1)}_t,Y^{(2)}_t)$ can be used as integrator in the pathwise It\^o formula. 

Now let $T>0$ be given. Then by \eqref{Ct eq} there exists $\eps>0$ such that $Y^{(2)}_t\ge\eps$ for $0\le t\le T$. Let $f\in C^2(\bR^2)$ be a function such that
$f(y^1,y^2)= y^1/y^2$ for $y^2\ge \eps/2$. Theorem~\ref{FoellmerThm2} then yields that
\begin{equation}\label{Y Ito eq}
\begin{split}
f(\bm Y_t)-f(\bm Y_0)&=\int_0^t\nabla f(\bm Y_s)\,d\bm Y_s+\frac12\int_0^tf_{y^1y^1}(\bm Y_s)\,d[Y^{(1)}]_s\\&\qquad +\frac12\int_0^tf_{y^2y^2}(\bm Y_s)\,d[Y^{(2)}]_s+\int_0^tf_{y^1y^2}(\bm Y_s)\,d[Y^{(1)},Y^{(2)}]_s.
\end{split}\end{equation}
Applying Theorem~\ref{associativity thm} with $\nu=2$, $\bm \eta_t=\nabla f(\bm Y_s)$, $\xi^{(\ell)}_t:={m_sY^{(\ell)}_s}/{S_s}$, $d=1$, and $ X=S$ yields that the It\^o integral above is given by
\begin{eqnarray*}
\int_0^t\nabla f(\bm Y_s)\,d\bm Y_s&=&\int_0^t\bigg(f_{y^1}(\bm Y_s)\frac{m_sY^{(1)}_s}{S_s}+f_{y^2}(\bm Y_s)\frac{m_sY^{(2)}_s}{S_s}\bigg)\,dS_s.
\end{eqnarray*}
Since $f_{y^1}(\bm Y_s)=1/Y^{(2)}_t$ and $f_{y^2}(\bm Y_s)=-Y^{(1)}_t/(Y^{(2)}_t)^2$, we see that the integrand of the right-hand integral vanishes. Hence $\int_0^t\nabla f(\bm Y_s)\,d\bm Y_s=0$ for $0\le t\le T$. Moreover, $f_{y^1y^1}=0$ and so also the second integral on the right-hand side of \eqref{Y Ito eq} vanishes. Finally, one easily shows with \eqref{Yi q var eq} and the associativity of the Stieltjes integral that the remaining two integrals on the right-hand side of \eqref{Y Ito eq} add up to zero. Thus, $Y^{(1)}_t/Y^{(2)}_t=f(\bm Y_t)=f(\bm Y_0)=Y^{(1)}_0/Y^{(2)}_0=1$ and so $\wt C_t=C_t$ for all $t\in[0,T]$. Therefore the uniqueness of the DPPI strategy follows.
\end{proof}

\begin{proof}[Proof of Theorem~\ref{CPPIProp}] Take $T>0$ and let $\eps>0$ 
be such that $S_t\ge\eps$ for $0\le t\le T$. Then we take $f\in C^2(\bR)$ such that $f(x)=\log x$ for $x\ge\eps/2$. When $m$ is constant, an application of the pathwise It\^o formula to $mf(S_t)$ yields that
$$ \int_0^t\frac{m}{S_s}\,d S_s=m\log S_t-m\log S_0+\frac m2\int_0^t\frac1{S_s^2}\,d[S]_s.
$$
Moreover, \cite[Proposition 2.2.10]{Sondermann} yields that 
$$[\log S]_t=\int_0^t\frac1{S_s^2}\,d[S]_s.
$$
Hence, formula \eqref{Ct eq} becomes
\begin{eqnarray*}
C_t&=&(1-\alpha)V_0\exp\bigg(m\log S_t-m\log S_0-\frac{m(m-1)}2\int_0^t\frac1{S_s^2}\,d[S]_s+(1-m)\int_0^tr_s\,ds\bigg)\\
&=&(1-\alpha)V_0\bigg(\frac{S_t}{S_0}\bigg)^me^{-\frac12m(m-1)[\log S]_t}B_t^{1-m}.
\end{eqnarray*}
This concludes the proof.\end{proof}

\bibliography{Bib}{}

\begin{thebibliography}{10}

\bibitem{AliprantisBorder}
C.~D. Aliprantis and K.~C. Border.
\newblock {\em Infinite-dimensional analysis. A hitchhiker's guide}.
\newblock Springer-Verlag, Berlin, second edition, 1999.

\bibitem{Balder}
S.~Balder, M.~Brandl, and A.~Mahayni.
\newblock Effectiveness of {CPPI} strategies under discrete-time trading.
\newblock {\em J. Econom. Dynam. Control}, 33(1):204--220, 2009.

\bibitem{Benderetal1}
C.~Bender, T.~Sottinen, and E.~Valkeila.
\newblock Pricing by hedging and no-arbitrage beyond semimartingales.
\newblock {\em Finance Stoch.}, 12(4):441--468, 2008.

\bibitem{Benderetal}
C.~Bender, T.~Sottinen, and E.~Valkeila.
\newblock Fractional processes as models in stochastic finance.
\newblock In G.~Di~Nunno and B.~{\O}ksendal, editors, {\em Advanced
  mathematical methods for finance}, pages 75--103. Springer, Heidelberg, 2011.

\bibitem{BickWillinger}
A.~Bick and W.~Willinger.
\newblock Dynamic spanning without probabilities.
\newblock {\em Stochastic Process. Appl.}, 50(2):349--374, 1994.

\bibitem{BlackJones}
F.~Black and R.~C. Jones.
\newblock Simplifying portfolio insurance.
\newblock {\em The Journal of Portfolio Management}, 14(1):48--51, 1987.

\bibitem{BlackPerold}
F.~Black and A.~Perold.
\newblock Theory of constant proportion portfolio insurance.
\newblock {\em Journal of Economic Dynamics and Control}, 16(3):403--426, 1992.

\bibitem{BrownHobsonRogers}
H.~Brown, D.~Hobson, and L.~C.~G. Rogers.
\newblock Robust hedging of barrier options.
\newblock {\em Math. Finance}, 11(3):285--314, 2001.

\bibitem{Buehler}
H.~B{\"u}hler.
\newblock Consistent variance curve models.
\newblock {\em Finance Stoch.}, 10(2):178--203, 2006.

\bibitem{Cheridito}
P.~Cheridito.
\newblock Arbitrage in fractional {B}rownian motion models.
\newblock {\em Finance Stoch.}, 7(4):533--553, 2003.

\bibitem{Cont}
R.~Cont.
\newblock Model uncertainty and its impact on the pricing of derivative
  instruments.
\newblock {\em Math. Finance}, 16(3):519--547, 2006.

\bibitem{ContTankovCPPI}
R.~Cont and P.~Tankov.
\newblock Constant proportion portfolio insurance in the presence of jumps in
  asset prices.
\newblock {\em Mathematical Finance}, 19(3):379--401, 2009.

\bibitem{CoxObloj}
A.~M.~G. Cox and J.~Ob{\l}{\'o}j.
\newblock Robust hedging of double touch barrier options.
\newblock {\em SIAM J. Financial Math.}, 2:141--182, 2011.

\bibitem{DavisRavalObloij}
M.~Davis, J.~Ob{\l}{\'o}j, and V.~Raval.
\newblock Arbitrage bounds for prices of weighted variance swaps.
\newblock {\em {\rm To appear in} Mathematical Finance}, 2013.

\bibitem{FoellmerIto}
H.~F{\"o}llmer.
\newblock Calcul d'{I}t\^o sans probabilit\'es.
\newblock In {\em Seminar on {P}robability, {XV} ({U}niv. {S}trasbourg,
  {S}trasbourg, 1979/1980) ({F}rench)}, volume 850 of {\em Lecture Notes in
  Math.}, pages 143--150. Springer, Berlin, 1981.

\bibitem{FoellmerECM}
H.~F{\"o}llmer.
\newblock Probabilistic aspects of financial risk.
\newblock In {\em European {C}ongress of {M}athematics, {V}ol. {I}
  ({B}arcelona, 2000)}, volume 201 of {\em Progr. Math.}, pages 21--36.
  Birkh\"auser, Basel, 2001.

\bibitem{FoellmerSchied}
H.~F{\"o}llmer and A.~Schied.
\newblock {\em Stochastic finance. An introduction in discrete time}.
\newblock Walter de Gruyter \& Co., Berlin, 3rd revised and extended edition,
  2011.

\bibitem{FoellmerSchiedBernoulli}
H.~F{\"o}llmer and A.~Schied.
\newblock Probabilistic aspects of finance.
\newblock {\em Bernoulli}, 19:1306--1326, 2013.
\newblock Special issue in celebration of the 300th anniversary of the
  publication of Jacob Bernoulli\rq s Ars Conjectandi.

\bibitem{GilboaSchmeidler}
I.~Gilboa and D.~Schmeidler.
\newblock Maxmin expected utility with nonunique prior.
\newblock {\em J. Math. Econom.}, 18(2):141--153, 1989.

\bibitem{HansenSargent}
L.~P. Hansen and T.~J. Sargent.
\newblock {\em Robustness}.
\newblock Princeton university press, 2011.

\bibitem{Knight}
F.~Knight.
\newblock {\em Risk, uncertainty, and profit}.
\newblock Houghton Mifflin, Boston, 1921.

\bibitem{Maccheronietal}
F.~Maccheroni, M.~Marinacci, and A.~Rustichini.
\newblock Ambiguity aversion, robustness, and the variational representation of
  preferences.
\newblock {\em Econometrica}, 74(6):1447--1498, 2006.

\bibitem{Paulot}
L.~Paulot and X.~Lacroze.
\newblock One-dimensional pricing of {CPPI}.
\newblock {\em Appl. Math. Finance}, 18(3):207--225, 2011.

\bibitem{Perold}
A.~F. Perold.
\newblock Constant proportion portfolio insurance.
\newblock {\em Harvard Business School}, 1986.

\bibitem{Salopek}
D.~M. Salopek.
\newblock Tolerance to arbitrage.
\newblock {\em Stochastic Process. Appl.}, 76(2):217--230, 1998.

\bibitem{SchiedStadje}
A.~Schied and M.~Stadje.
\newblock Robustness of delta hedging for path-dependent options in local
  volatility models.
\newblock {\em J. Appl. Probab.}, 44(4):865--879, 2007.

\bibitem{Shiryaev}
A.~N. Shiryaev.
\newblock {\em Essentials of stochastic finance. Facts, models, theory},
  volume~3 of {\em Advanced Series on Statistical Science \& Applied
  Probability}.
\newblock World Scientific Publishing Co. Inc., River Edge, NJ, 1999.
\newblock Translated from the Russian manuscript by N. Kruzhilin.

\bibitem{Sondermann}
D.~Sondermann.
\newblock {\em Introduction to stochastic calculus for finance. A new didactic
  approach}, volume 579 of {\em Lecture Notes in Economics and Mathematical
  Systems}.
\newblock Springer-Verlag, Berlin, 2006.

\bibitem{Widder}
D.~V. Widder.
\newblock {\em The {L}aplace {T}ransform}.
\newblock Princeton Mathematical Series, v. 6. Princeton University Press,
  Princeton, N. J., 1941.

\end{thebibliography}
\bibliographystyle{abbrv}

\end{document}